\documentclass{article}
\usepackage[left=3.5cm,right=3.5cm,top=2.5cm,bottom=2.5cm]{geometry}
\usepackage[english]{babel}
\usepackage[utf8]{inputenc}
\usepackage{amsmath,amssymb,amsfonts,amsthm}
\usepackage{enumerate}
\usepackage{fancyhdr}
\usepackage{color}
\usepackage{graphicx}
\usepackage{hyperref}
\usepackage[version=3]{mhchem}
\usepackage{mathtools}

\usepackage[ddmmyyyy]{datetime}
\usepackage{blkarray}
\usepackage{mathtools}

\newcommand{\N}{\mathbb{N}}

\newcommand{\R}{\mathbb{R}}

\newcommand{\Rp}{\mathbb{R}_{>0}}
\newcommand{\Rnn}{\mathbb{R}_{\geq 0}}

\DeclareMathOperator{\im}{im} 
\renewcommand{\k}{{\kappa}}
\newcommand{\C}{\mathcal{C}} 

\newcommand{\mX}{\mathcal{X}}   
\newcommand{\mC}{\mathcal{C}}    
\newcommand{\mR}{\mathcal{R}} 
\newcommand{\mL}{\mathcal{L}}

\newcommand{\st}{\, | \, }

\newtheorem{theorem}{Theorem} 
\newtheorem{lemma}[theorem]{Lemma} 
 
\newtheorem{corollary}[theorem]{Corollary} 
\theoremstyle{definition}
\newtheorem{remark}[theorem]{Remark} 
\newtheorem{example}[theorem]{Example}
\newtheorem{definition}[theorem]{Definition}

\newenvironment{enumerate*}[1][{}]{\begin{itemize}\renewcommand{\labelitemi}{#1}}{\end{itemize}}

\newcommand{\edit}[1]{{\color{black}#1}}

\title{Addition of flow reactions preserving multistationarity and bistability}

\author{Daniele Cappelletti\footnotemark[1] \and Elisenda Feliu\footnotemark[2] \and Carsten Wiuf\footnotemark[2]}

\date{\today}

\begin{document}

\maketitle

\renewcommand{\thefootnote}{\fnsymbol{footnote}}
\footnotetext[1]{Department of Biosystems Science and Engineering, ETH Zurich, Mattenstrasse 26, 4058 Basel, Switzerland.}
\footnotetext[2]{Department of Mathematical Sciences, University of Copenhagen, Universitetsparken 5, 2100 Copenhagen, Denmark.}
\renewcommand{\thefootnote}{\arabic{footnote}}

\begin{abstract}
We consider the question whether a chemical reaction network preserves the number and stability of its positive steady states upon inclusion of inflow and outflow reactions. Often a model of a reaction network is presented without inflows and outflows, while in fact some of the species might be degraded or leaked to the environment, or be synthesized or transported into the system. We provide a sufficient and easy-to-check criterion based on the stoichiometry of the reaction network alone and discuss examples from systems biology.

\end{abstract}

\pagestyle{myheadings}
\thispagestyle{plain}
\markboth{D. Cappelletti, E. Feliu, C. Wiuf}{Addition of flow reactions preserving multistationarity and bistability}

\section{Introduction}

Bistability and  multistationarity in general are considered  important biological mechanisms, providing explanations for co-existence of differentiable phenotypes and switch-like behaviour  \cite{Legewie:2005hw,palani}.
The question of whether bistability is present in a given system therefore arises naturally in many contexts \cite{Liu:2011p964,Markevich-mapk,harrington-feliu,Nguyen:2011p975,Eissing2004}. However, it is not straightforward to decide \emph{a priori} whether this is the case.

The objects of interest here are reaction networks describing  the evolution of  concentrations of chemical species over time, and modeled by means of systems of autonomous Ordinary Differential Equations (ODEs). 
Such an ODE system typically admits linear conservation laws, that is, linear first integrals, due to the reactions alone, independently of the kinetics. The first integrals restrict the dynamics of the ODE system to the so-called stoichiometric compatibility classes, and questions about  the existence of multiple steady states and their stability properties are to be addressed relatively to each class. Furthermore, the ODE system depends  on (potentially many) unknown parameters, which adds to the  difficulty of the problem as the  number and stability of the steady states must be investigated for general parameters. In particular, a reaction network is said to be multistationary, or bistable, if that is the case for some choice of parameter values. 

A successful strategy to determine whether a network is multistationary or bistable  is the following. First the number and stability of the steady states of a reduced reaction network is  studied, and then these steady states  (and their properties) are ``lifted'' to the original network. As the ODE system arises from a reaction network, a substantial amount of recent work has focused on determining modifications of  the reactions that preserve properties at steady state upon lifting. Specifically, we consider two reaction networks $F$ and $G$, with respective kinetic rates, and aim to prove statements of the form
\begin{center}
``Provided (\dots),  if $F$ has $\ell$ positive/\edit{asymptotically stable/exponentially stable}/unstable \\ steady states for some  parameter choice, then   so does $G$ for some parameter choice.''
\end{center}

Perhaps,  the first  work in this direction is due to Craciun and Feinberg~\cite{craciun-feinberg}. They show that multistationarity is preserved if  a reaction of the form $X\ce{<=>} 0$ is added to a network $F$ for all species $X$ in $F$. 
Subsequent work by Joshi and Shiu~\cite{joshi-shiu-II} consider the case where $G$ is obtained from $F$ by adding reactions in such a way that the stoichiometric compatibility classes are preserved. They also consider the case of \emph{embedded networks}, where $G$ is obtained  from $F$ by adding species in a specific way. Feliu and Wiuf \cite{feliu:intermediates}  show that the number and properties of  the steady states of a network are preserved upon  the addition of \emph{intermediate species}. Banaji and Pantea \cite{banaji-pantea-inheritance} introduce additional operations that preserve  steady states properties.

Here we revisit the situation in which reactions of the form
\[ 0\ce{->} X,\qquad X\ce{->} 0,\] 
called \emph{ inflow}  and \emph{outflow} reactions, respectively, and jointly \emph{flow} reactions, are added to $F$, where  $X$ is a species already in $F$. 
The addition of such reactions destroys the linear first integrals involving the concentration of $X$ and hence increases the actual dimension of the ODE system. Some recent methods to count the number of steady states and to determine their stability 
rely heavily on the existence of a parametrization of the steady state manifold, not restricted to a particular stoichiometric compatibility class \cite{FeliuPlos,torres:bistability,dickenstein:regions,conradi-mincheva,PerezMillan}. Finding such parametrizations often requires sufficient freedom, which comes from the codimension of the stoichiometric compatibility classes. As a consequence, the direct determination of the number and stability of the steady states of a reaction network with flow reactions is much harder than for networks without them. 

As mentioned above, if both flow reactions are added for \emph{all} species of $F$, then statements of the desired type can be obtained (lifting of steady states appears first in \cite{craciun-feinberg}; stability in \cite{banaji-pantea-inheritance}).  Hence, the network without flow reactions provides information about the network with all flow reactions. 
However, from a biochemical or metabolic point of view, the network with all flow reactions makes generally little sense. Here, $0\ce{->}X$ often represents synthesis of $X$ or transport of $X$ into a compartment, while $X\ce{->} 0$ represents degradation\edit{, dilution} or transport out of a compartment. Modeling, for example, the inflow of a protein complex, such as a kinase-substrate complex, is in general not meaningful. Indeed, in realistic models inflow and outflow reactions are only considered for selected \edit{(possibly different) sets of} species. \edit{As an example, in biomolecular models it is frequently the case that an outflow for all chemical species is present, but only some species participate in an inflow reaction.} 

The discussion raises the following question: 
\begin{center}
What \edit{inflow and outflow} reactions can be added to a reaction network while preserving \\ the number and stability \edit{properties} of the steady states?
\end{center}
(in the sense discussed above).
We should not expect that an arbitrary  selection of flow reactions is allowed. As a simple illustration, consider the reaction network $X_1\ce{<=>} X_2$ with mass-action kinetics. It has one positive steady state in each stoichiometric compatibility class (for any choice of reaction rate constants). By adding the outflow reaction $X_2\ce{->} 0$, the network has no positive steady states for any choice of reaction rate constants. 

In this paper we give a sufficient and easy-to-check criterion based on  stoichiometry alone to decide on the question above. In particular, the criterion \edit{gives sufficient conditions on the choice of flow reactions as follows}. 
\edit{We ask for the existence of a basis of conservation laws such that, for every conservation law, the species corresponding to non-zero coefficients are either all in added outflow reactions or none is.  Each species in an inflow reaction must either also be in an outflow reaction, or not conserved in the original network. Finally, we require the restriction of the conservation laws to the subset of inflow species to generate the same conserved quantities as the original conservation laws. These statements are properly formalised in Section~\ref{sec:thm}. } We give various examples and illustrate, also by example, that our conditions are not necessary.

\section{Reaction Networks and Steady States}\label{sec:reactions}

Let $\Rp$ and $\Rnn$ denote the set of positive real numbers and the set of non-negative real numbers, respectively. Let $\N$ be the set of non-negative integers. $\langle v_1,\ldots,v_k\rangle$ denotes the linear span of the vectors $v_1,\ldots,v_k$ (in some vector space).

\medskip

A \emph{reaction network} $F=(\mC,\mR)$ on a non-empty set $\mX=\{X_1,\dots,X_{n}\}$ is a directed graph without self-edges whose nodes are linear combinations in $\mX$ with non-negative integer coefficients. The elements $y\in \mC$ are  \emph{complexes} and  of the form $y=\sum_{i=1}^n \lambda_i X_i$ with $\lambda_i\in \N$, $i=1,\dots,n$. 
The elements in $\mR$ are  \emph{reactions}. It is assumed that there are no isolated nodes, that is, every complex is part of a reaction. We further consider the set of reactions to be  ordered and let $m$ be its cardinality. 

\edit{Note that $\mC$ might include $0$, the so-called zero complex. }
For $X\in\mX$, recall from the introduction that the reactions 
\[X\rightarrow 0\quad \textrm{ and }\quad 0\to X\]
 are outflow and inflow reactions, respectively, and jointly  referred to as flow reactions. 
\edit{If at least one of the above flow reactions is present, then} $X$ is said to be a flow species. 

We identify  the species $X_i$  with the $i$-th canonical vector of $\R^{n}$ with $1$ in the $i$-th position and zeroes elsewhere. Hence, each complex $y\in \mC$ is a vector in $\R^n$. 
 The \emph{stoichiometric matrix} $N\in \R^{n\times m}$ of $F$ is the matrix whose $j$-th column is the vector $y'-y$, where $y\rightarrow y'$ is the $j$-th reaction. 
  In particular, the $(i,j)$-th entry of $N$ encodes the net production of species $X_i$ in the $j$-th reaction.
The \emph{stoichiometric subspace}  $S\subseteq \R^{n}$ of $F$ is the span of the column vectors of $N$, 
   \[ S:=\im(N)\subseteq \R^{n}.\] 
    We denote by   $s$ the dimension of $S$ (that is, the rank of $N$) and by $d$ the dimension of the orthogonal complement subspace $S^\perp$ of $S$. Hence  $s+d=n$.  

 The species concentrations  change over time  as a consequence of the reactions taking place. To describe the time evolution  we introduce a kinetics and the species-formation rate function. 
A \emph{kinetics} for   $F$   is a $\mC^1$-function 
\[ K \colon \Omega \longrightarrow \R^{m}_{\geq 0},\]
where $\Rp^{n} \subseteq \Omega \subseteq \Rnn^{n}$ (differentiability is with respect to the open set $\R^n_{>0}$).
The entry $K_j(x)$ is   called the \emph{rate of the $j$-th reaction}. 
A common choice of kinetics is   \emph{mass-action kinetics}, with $ \Omega=\R^{n}_{\geq 0}$ and 
\[ K_j(x)= \k_j \prod_{i=1}^{n} x_i^{y_i},\qquad\textrm{if }y\rightarrow y'\textrm{ is the }j\textrm{-th reaction},\]
where $\k_j>0$ is the \emph{reaction rate constant} of the reaction. 
Under this kinetics,  the reactions are usually labeled with the reaction rate constants.

The \emph{species-formation rate function} of $F$ with kinetics $K$ is the map $f\colon \Omega \rightarrow S\subseteq \R^{n}$ defined by
\[ f(x) :=N K(x).\] 
The dynamics of the species concentrations  of the network $F$ with kinetics $K$  is described by a set of ODEs given by the species-formation rate function: 
\begin{align} \label{eq:ode}
\dot{x} &=f(x), \qquad x\in \Omega,
\end{align}
where $\dot{x}=\dot{x}(t)$ denotes the derivative of $x(t)$ with respect to time $t$. 

If $K_j(x)$ vanishes whenever $x_i=0$ for $i$ an index for which $N_{ij}$ is negative (as it is the case for \edit{common kinetics including} mass-action kinetics), then both $  \Omega$ and $\R^n_{>0}$ are forward invariant by the solutions of \eqref{eq:ode} \edit{\cite{volpert}.} 
Additionally, as $\im(f) \subseteq S$, a solution of \eqref{eq:ode} is confined to an invariant linear space of the form $x_0+S$, where $x_0\in \Omega$ is the initial point of the solution. The   set 
$(x_0+S)\cap \Omega$ is called a \emph{stoichiometric compatibility class}. 
Given a matrix $W\in \R^{d\times n}$ whose rows form a basis of $S^\perp$,  
\edit{the elements $x\in (x_0+S)\cap \Omega$ are the solutions to the equation $W x = W x_0$ in $\Omega$. Hence, }
the stoichiometric compatibility classes might be parametrized by $T=(T_1,\ldots,T_{d})\in W(\Omega) \subseteq \R^{d}$ as follows:
\begin{equation}\label{eq:stoichclass}
 \mL_{W,T} := \big \{x\in \Omega \st W x = T\big \}.
\end{equation}
\edit{If $W$ is fixed, }given $x_0\in \Omega$, $x_0\in \mL_{W,T}$ \edit{only if $T=W x_0$}. The vector $T$ is commonly referred to as the vector of \edit{\emph{conserved quantities}} and any relation $\sum_{i=1}^n \omega_i x_i=T$, with $\omega\in S^\perp$, is called a \emph{conservation law}. Consequently, we call a matrix $W$ whose rows form a basis of $S^\perp$ a \emph{matrix of conservation laws}.

\begin{definition}
Given a vector $u$ in $\R^n$ we define its \emph{support} to be the subset of species $X_i\in \mX$ where $u_i\neq 0$. 
We say a species $X_i$ is \emph{non-conserved} if it is not in the support of any vector in $S^\perp$, that is, the $i$-th canonical vector of $\R^n$  belongs to $S$. 
\end{definition}

Note that $X_i$ is non-conserved if and only if the $i$-th column of a matrix of conservation laws is zero.

\begin{example}
\label{ex:main} 
Consider the reaction network  with mass-action kinetics, given by 
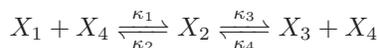
\[ X_1+X_4 \ce{<=>[\k_1][\k_2]} X_2  \ce{<=>[\k_3][\k_4]} X_3+X_4\]
where $\k_1,\k_2,\k_3,\k_4$ are positive constants.  
The dynamics of the species concentrations are described by the following ODE system:
\begin{align*}
\dot{x}_1 &= -\k_1 x_1x_4 + \k_2 x_2  & \dot{x}_2 &= \k_1 x_1x_4 -(\k_2+\k_3) x_2 + \k_4 x_3x_4   \\
\dot{x}_3 &= \k_3 x_2 - \k_4 x_3x_4  & \dot{x}_4 &= -\k_1 x_1x_4 + (\k_2+\k_3) x_2  - \k_4 x_3x_4.
\end{align*}
The stoichiometric matrix and the stoichiometric subspace are
\[  N= \begin{pmatrix}
 -1 & \phantom{-}1 & \phantom{-}0 & \phantom{-}0 \\
\phantom{-}1 & -1 & -1 & \phantom{-}1 \\
  \phantom{-}0 & \phantom{-}0 & \phantom{-}1 & -1 \\
 -1 & \phantom{-}1 & \phantom{-}1 & -1 
 \end{pmatrix}, \qquad S = \Big\langle (-1,1,0,-1),(0,-1,1,1)\Big\rangle.\]
 Therefore $n=4,m=4,s=2,$ and $d=2$.
  \end{example}

The \emph{steady states} of a network $F$ with kinetics $K$ are the solutions to the system of equations
\begin{equation}\label{eq:ss}
\edit{f(x)=N K(x)=0, \qquad x\in \Omega,}
\end{equation}
which we refer to as the \emph{steady state equations}.
If $N'$ is any matrix such that $\ker(N)=\ker(N')$, then the steady states of the network are 
precisely the  solutions to $N' K(x)=0$. In particular, if the rank of $N$ is $s$, then there exists a matrix $N'\in \R^{s\times m}$ of maximal rank such that $\ker(N)=\ker(N')$. Therefore, the system of $n$ equations \eqref{eq:ss}  can always be reduced to an equivalent system of $s$ equations.
 
As we consider the dynamics of the system confined to  the stoichiometric compatibility classes, questions about the number, stability or other properties of the steady states will be addressed relatively to a given stoichiometric compatibility class.
Specifically, the steady states of a network with kinetics $K$  in a stoichiometric compatibility class $\mL_{W,T}$ are the   solutions to the $n$ equations
\[ N' K(x)=0 \quad  \text{and} \quad  W x = T,\qquad x\in \Omega,\] 
where $N'\in \R^{s\times m}$ is such that $\ker(N)=\ker(N')$. 
We define accordingly  
\begin{equation}\label{eq:Phi}
\Phi(x):=  \big(N' K(x), Wx - T  \big) \in \R^{s} \times \edit{\R^{d}} \equiv \R^{n},\qquad x\in \Omega,
\end{equation}
such that  the steady states in   $\mL_{W,T}$ are the solutions to $\Phi(x)=0$. 
We say a network with kinetics $K$ is \emph{multistationary} if there exists a stoichiometric compatibility class containing at least two positive steady states.

Recall that the \emph{Jacobian} $J_f(x^*)$ of a $\mC^1$-map $f\colon U\subseteq \R^n\rightarrow \R^n$ evaluated at $x^*\in \text{int}(U)$ is the $n\times n$ matrix such that the $(i,j)$-th entry is
$\frac{\partial f_i }{\partial x_j}(x^*)$.  
We say that a steady state $x^*$    is \emph{non-degenerate}  if  $J_{f}(x^*)$ is non-singular on $S$, i.e. $\ker(J_{f}(x^*))\cap S=\{0\}$. The following  lemma is  proved in \cite{wiuf-feliu}.

\begin{lemma}\label{lemma:nondeg}
A steady state $x^*\in \Omega$ in the stoichiometric compatibility class  $\mL_{W,T}$  is non-degenerate if  and only if the Jacobian $J_\Phi(x^*)$ of $\Phi$ defined in \eqref{eq:Phi} evaluated at $x^*$ is non-singular on $\R^{n}$, that is,  if and only if $\det(J_{\Phi}(x^*))\neq 0$.
\end{lemma}

Given a steady state $x^*$, the eigenvalues of $J_f(x^*)$ convey information on the \edit{local asymptotic stability} of $x^*$ \emph{relative to the stoichiometric compatibility class} it belongs to. 
As the rank of $J_f(x^*)$ is at most $s$, $0$ is an eigenvalue of $J_f(x^*)$  with multiplicity at least $d$, and the multiplicity is exactly $d$ if the steady state is non-degenerate.
If  $J_f(x^*)$ has $s$ eigenvalues with negative real part (counted with multiplicity), then $x^*$ is said to be \emph{exponentially stable}, and is in particular asymptotically stable relative to the stoichiometric compatibility class. If at least one of the eigenvalues of  $J_f(x^*)$ has positive real part, then  $x^*$ is said to be \emph{exponentially unstable} (this is not a standard term in the literature but used here for convenience), and is in particular unstable. If $J_f(x^*)$ has $s$ eigenvalues with non-vanishing real part, then the steady state is said to be \emph{hyperbolic}.

\section{Partially open extensions and lifting steady states}\label{sec:thm}
 
In this section we compare two reaction networks $F$ and $G$ with respective kinetics, such that $G$ is obtained from $F$ by adding some inflow and outflow reactions.  
In what follows the objects defined in the previous section, namely $N, S, \mR,W,K$, are indexed by subscripts  $F$ and $G$ indicating the network they are associated with. 
 The  following definition is inspired by \cite{craciun-feinberg-semiopen,craciun-feinberg}.

 \begin{definition}\label{def:partially}
  Consider a reaction network $F$ with kinetics $K_F$ and species set $\mX$ of cardinality $n$.  
\begin{itemize}
\item We say that a reaction network $G$ with kinetics $K_G$ is a \emph{partially open extension of $F$ with respect to the inflow set $\mX^\iota\subseteq \mX$ and the outflow set $\mX^o \subseteq \mX$ }, if the species set of $G$ is $\mX$, the set $\mR_G$ decomposes as a disjoint union
\[ \mR_G = \mR_F \sqcup \{ 0 \ce{->} X\}_{X \in \mX^\iota} \sqcup \{ X \ce{->} 0\}_{X\in \mX^o},\] 
and $K_G$ agrees with $K_F$ for the reactions in $\mR_F$ and \edit{is of mass-action type for the reactions in $\mR_G\setminus \mR_F$.}
\item If additionally $G$ has an inflow and outflow reaction for all species in $\mX$, 
 then $G$ is called a \emph{fully open extension of $F$}.
\end{itemize}
  \end{definition}

Note that $F$ might  have flow reactions for species not in $\mX^o \cup \mX^\iota$.  
\edit{We emphasize that $K_F$ is any kinetics, not necessarily mass-action, while the rate of the inflow  and outflow reactions of species in $\mX^\iota$ and $\mX^o$ respectively is mass-action. }
 Our goal is to study for what sets $\mX^\iota,\mX^o$,  the number and stability of the positive steady states of $F$ extend to $G$ after appropriately choosing  reaction rate constants for the added flow reactions. 
By \cite{craciun-feinberg}, if $F$ has multiple positive non-degenerate steady states, then so does the fully open extension, provided the reaction rate constants of the  flow reactions of $G$ that are not in $F$ are chosen small enough. By \cite{banaji-pantea-inheritance}, the maximal number of exponentially stable positive steady states of $F$ within a stoichiometric compatibility class  is also a lower bound of the maximal number of exponentially stable positive steady states $G$ admits for arbitrary reaction rate constants of the flow reactions. 
Here we relax the condition that \emph{all missing} inflow and outflow reactions must be added to preserve these characteristics.

\medskip
A key ingredient of the main theorem below is to understand the image of the positive orthant by a matrix of conservation laws $W$, which is a polyhedral cone. Thus, in preparation for the main theorem, 
we discuss a well-known property of polyhedral cones in $\R^n$.
Given a matrix $M\in \R^{r\times n}$ with $r\leq n$ and of rank $r$, let $\mathcal{C}(M)$ be the polyhedral cone generated by the columns of $M$: if  $m^{(i)}$ denotes the $i$-th column of $M$, then 
\[ \mathcal{C}(M) = \left\{ \sum_{i=1}^n \lambda_i m^{(i)} \in \R^r \st \lambda_1,\dots,\lambda_n \geq 0    \right\} = \big\{ M \lambda \st \lambda \in \R^n_{\geq 0}\big\} .\]
We let $\mathcal{C}^o(M) $ denote the corresponding open cone, obtained by imposing all $\lambda_i$ to be positive.

For a subset $I\subseteq \{1,\dots,n\}$, we let $M^I$ denote the submatrix of $M$ given by the columns with index in $I$. We  say that $M^{I}$ \emph{generates} $\mathcal{C}^o(M) $ if $\mathcal{C}^o(M)= \mathcal{C}^o(M^{I})$.

\begin{lemma}\label{lem:cone}
Let $M,M^*\in \R^{r\times n}$ be of   rank $r\leq n$  and such that $M^*=P M$ with $P\in \R^{r\times r}$ invertible. Given $I\subseteq \{1,\dots,n\}$,  $M^I$ generates $\mathcal{C}^o(M) $ if and only if $M^{*I}$ generates $\mathcal{C}^o(M^*)$.
\end{lemma}
\begin{proof}
Assume $M^I$ generates  $\mathcal{C}^o(M)$, such that the images of $\R^n_{> 0}$ by $M$ and by  $M^I$ agree. Now we have 
\begin{align*}
\mathcal{C}^o(M^*)&= \{ M^* \lambda  \st \lambda \in \R^n_{> 0}\} =  \{ P (M \lambda)  \st \lambda \in \R^n_{>  0}\}
=  \{ P (M^{I} \lambda)  \st \lambda \in \R^{|I|}_{> 0}\} \\ & =   \{ M^{*I} \lambda  \st \lambda \in \R^{|I|}_{> 0}\} =
\mathcal{C}^o(M^{*I}),
\end{align*}
where $|I|$ denotes the cardinality of $I$.
By symmetry of the argument, this concludes the proof. 
\end{proof}

Under the hypothesis of Lemma~\ref{lem:cone}, the rows of $M$ and $M^*$ generate the same vector subspace of $\R^n$. 
In view of Lemma~\ref{lem:cone}, the following definition is consistent and independent of the choice of $W$.

\begin{definition}\label{def:generatesW}
Consider a reaction network $F$  with species set $\mX$ of cardinality $n$.
\begin{itemize}
\item For any subset $\mX_1\subseteq \mX,$  the vector subspace of conservation laws with support in $\mX_1$ is defined as
\[ S_{F,\mX_1}^\perp =   S_F^\perp \cap  \big \langle e_{i} \st  X_i\in  \mX_1\big\rangle \edit{\subseteq \R^n}, \]
\edit{where $e_i$ denotes the $i$-th canonical vector of $\R^n$. }
\item Given subsets $\mX_1\subseteq \mX_2 \subseteq \mX$, let $I$ be the index set of the species in $\mX_1$, and $W_2$ be a matrix whose rows form a basis of the conservation laws with support in $\mX_2$, $S_{F,\mX_2}^\perp$. 
Then $\mX_1$ is said to  \emph{generate}  $S_{F,\mX_2}^\perp$ if $W_2^{I}$ generates $\mathcal{C}^o(W_2)$.
\end{itemize}
\end{definition}

With the notation in Definition~\ref{def:generatesW}, let $n_1, n_2$ be the cardinality of $\mX_1,\mX_2$, respectively and $W'$ the submatrix of $W_2$ given by the columns corresponding to the species in $\mX_2$. Note that the other columns of  $W_2$ are zero. Then Definition~\ref{def:generatesW} is simply saying that the image of the positive orthant $\R_{>0}^{n_2}$  by $W'$ agrees with the image of the possibly lower dimensional orthant  $\{0\}^{n_2-n_1}\times \R_{>0}^{n_1}$  by $W'$ (after appropriate reordering of the variables). Note that $\mX_1$ needs to have at least as many species as rows of $W_2$. 
Furthermore, recall that the image of the positive orthant by a matrix of conservation laws determines the possible values of the vector of \edit{conserved quantities} for positive (steady) states of the system.

 \begin{lemma}\label{lem:directsum}
Consider a reaction network $F$ with kinetics $K_F$ and species set $\mX$ of cardinality $n$.
Let $G$ be a partially open extension of $F$ with \edit{inflow set $\mX^\iota$ and outflow set $\mX^o$} as in Definition~\ref{def:partially}. 
Then 
\[ S_G^\perp = S_{F,\mX\setminus (\mX^o \cup \mX^\iota)}^\perp= S_F^\perp \cap  \big \langle e_{i} \st  X_i\in \mX\setminus (\mX^o \cup \mX^\iota) \big\rangle,  \] 
that is, $S_G^\perp$ is the vector subspace of $S_F^\perp$ of vectors with support included in $\mX\setminus (\mX^o \cup \mX^\iota)$.
\end{lemma}
\begin{proof}
Recall that $S_G,S_F$ are vector subspaces of $\R^n$. 
By definition we have 
\[ S_G = S_F   + \big\langle e_{i} \st  X_i\in \mX^o \cup \mX^\iota \big\rangle,\]
where $e_i$ denotes the $i$-th canonical vector of $\R^n$. 
Hence 
\[ S_G^\perp = S_F^\perp \cap  \big\langle e_{i} \st  X_i\in \mX^o \cup \mX^\iota \big\rangle^\perp = 
S_{F}^\perp   \cap \big \langle e_{i} \st  X_i\in \mX\setminus (\mX^o \cup \mX^\iota)\big\rangle.\] 
\end{proof}

\edit{
We state now the main theorem, which gives three sufficient conditions (a)--(c) for which the desired lifting properties (i)--(iii) hold. The conditions are presented  in terms of $S^\perp_F$ and not in terms of matrices of conservation laws. We clarify how to check that the conditions hold using matrices of conservation laws  in Remark~\ref{rk:matrix} below. 
}

 \begin{theorem}\label{thm:main}
Consider a reaction network $F$ with kinetics $K_F$ and species set $\mX$ of cardinality $n$.
Let $G$ be a partially open extension of $F$ with respect to the inflow set  $\mX^\iota\subseteq \mX$ and the outflow set $\mX^o \subseteq \mX$, and with kinetics $K_G$. Let  $d_1$ be the dimension of \edit{$S_{F,\mX^o}^\perp$}.
Assume that
\begin{itemize}
\item[(a)] \edit{The species in the inflow set $\mX^\iota$  that are not in the outflow set $\mX^o$  are non-conserved species of $F$.}

\item[(b)] \edit{$S_F ^\perp$ decomposes as the direct sum of the vector subspace of vectors with support in $\mX^o$ and the vector subspace of vectors with support in $\mX\setminus \mX^o$: }
\[ \edit{S_F ^\perp = S_{F, \mX^o}^\perp \ \oplus \ S_{F,\mX\setminus \mX^o}^\perp}.
\] 
\item[(c)] \edit{The set }$\mX^\iota$ generates  $S_{F, \mX^o}^\perp$.
\end{itemize}
Then the following statements hold:
\begin{enumerate}[(i)]
\item If $F$ has at least $\ell$ positive non-degenerate steady states $c_1,\dots,c_\ell$ in one stoichiometric compatibility class $(x_0+S_F)\cap  \Omega$, then there exists a choice of reaction rate constants for the flow reactions of $G$ not in $F$ such that $G$ with this kinetics has at least $\ell$ positive non-degenerate steady states $c_1',\dots,c_\ell'$ in the  stoichiometric compatibility class $(x_0+S_G)\cap \Omega$.
\item With the appropriate numbering of steady states in (i), if $J_{f_F}(c_j)$ has at least $r_1$ eigenvalues with positive real part and $r_2$ eigenvalues with negative real part, then $J_{f_G}(c'_j)$  has \edit{at least}
 $r_1$ eigenvalues with positive real part and $r_2+d_1$ eigenvalues with negative real part. 
\item If $F$ has at least $\ell_1$ exponentially stable and $\ell_2$ exponentially unstable non-degenerate steady states in $(x_0+S_F)\cap \R^n_{> 0}$, then so does $G$ in $(x_0+S_G)\cap \R^n_{> 0}$, for a choice of reaction rate constants for the flow reactions of $G$ not in $F$ \edit{(which are endowed with mass-action kinetics by definition}).
\end{enumerate}
In particular, a choice of reaction rate constants of the flow reactions of $G$ not in $F$  such that \edit{(i)--(iii)} hold can be found as 
\begin{align*}
X_j &  \ce{->[\theta]} 0,  \quad X_j\in \mX^o, \qquad 0 \ce{->[\theta \widehat{x}_j]} X_j, \quad X_j\in \mX^\iota,
\end{align*}
with \edit{$\theta>0$ small enough, and where $\widehat{x}\in \R^n_{\geq 0}$ is any fixed vector with support $\mX^\iota$ satisfying $W_1 \widehat{x}= W_1 x_0$ for $W_1$ a matrix whose rows form a basis of $S_{F, \mX^o}^\perp$.}
\end{theorem}

We remark that in the notation of Theorem~\ref{thm:main}, $\ell_1+\ell_2$ might not be $\ell$, as not all $\ell$ steady states need to be hyperbolic. 
The proof of Theorem~\ref{thm:main} is given in Section~\ref{sec:proof} below and \edit{relies on the Implicit Function Theorem. To this end we construct a function that depends on $\theta$ that vanishes at the steady states of $F$ for $\theta=0$, and vanishes at  steady states of $G$ for $\theta>0$ with the kinetics in the statement of Theorem~\ref{thm:main}.  This is the standard approach underlying several of the known results lifting properties of steady states, e.g. \cite{feliu:intermediates,banaji-pantea-inheritance,craciun-feinberg}. The special aspect of the scenario considered in this paper is that  the addition of flow reactions typically decreases the dimension of the stoichiometric compatibility classes, and conditions (a)--(c) account for that. The role of conditions (a) and (b) of Theorem~\ref{thm:main}
is to allow for a separation of the conservation laws that involve the species in the outflow set, from those that do not and are not affected by this particular extension. \edit{According to Lemma~\ref{lem:directsum},} the latter are the conservation laws of $G$ 
as by assumption (a) we have $S_{F, \mX^o}^\perp =  S_{F, \mX^o\cup \mX^\iota}^\perp$ and  $ S_{F, \mX\setminus\mX^o}^\perp =  S_{F, \mX\setminus (\mX^o\cup \mX^\iota)}^\perp$. Using condition (c) and the choice of reaction rate constants, we relate the reaction rate constants of the inflow reactions to the conserved quantities of $F$ for the conservation laws with support in $\mX^o$ (which are not in $G$). }

\begin{remark}\label{rk:matrix}
The assumptions of Theorem~\ref{thm:main}  can easily be verified by considering a matrix of conservation laws  $W_F$ of $F$. 
In particular, assumption \edit{(a)} in Theorem~\ref{thm:main} says that for any species in \edit{$\mX^\iota$} for which the outflow reaction is not considered, then the corresponding column \edit{of any matrix of conservation laws} is identically zero.

Assumption \edit{(b)} of Theorem~\ref{thm:main} can be verified using Gauss reduction on $W_F$. Indeed, assuming $\mX$ is ordered such that the species in \edit{$\mX^o$} are the first species, then \edit{(b)} holds if and only if there is a matrix of conservation laws of $F$ of the form 
 \begin{equation}\label{eq:block_diag_form_inflows_outflows}
  W_F=\begin{pmatrix}
      W_1 & 0   \\
      0   & W_2 \\
     \end{pmatrix},
 \end{equation}
where the number of columns of $W_1$ is the cardinality of \edit{$\mX^o$}. It follows by Lemma~\ref{lem:directsum} \edit{and assumption (a)} that  $(0 \ W_2)$  is a matrix of conservation laws for $G$, that is, its rows form a basis of $S_G^\perp$. In particular $S_G$  is the kernel of $(0 \ W_2)$. 
The blocks $W_1$ and $W_2$ might  be empty.

Finally, assumption (c) can be verified by determining the rays $u_1,\dots,u_{k}$ of the cone $\mC(W_1)$. Then the columns of $W_1$ corresponding to $\mX^\iota$ must contain scalar multiples of all the vectors $u_1,\dots,u_{k}$.
Note that by assumption (a), assumption (c) holds for $\mX^\iota$ if and only if it holds for $\mX^{\iota}\cap \mX^o$.
\end{remark}
 
 \begin{remark}\label{rk:practical}
 Assume a set of species \edit{$\mX^o\subseteq\mX$} is given that satisfies assumption (b) of Theorem~\ref{thm:main}. Consider $\mX^\iota_1,\ldots,\mX^\iota_k\subseteq\mX^o$ the distinct smallest sets such that (c) holds (smallest in the sense that no proper subset of $\mX^\iota_j$ satisfies (c)).  Then $\mX^\iota$ is a set of inflow species for which (a) and (c) are satisfied if and only if $\mX^\iota_j\subseteq \mX^\iota$ for some $j=1,\ldots,k$ and further $\mX^\iota \setminus \mX^o$ consists of non-conserved species.
 If the rank of $W_1$ is $d_1$, then any set $\mX^\iota_j$ contains at least $d_1$ elements, providing a lower bound on the size of the set.   
\end{remark}

As an easy consequence of  Theorem~\ref{thm:main}, we recover two known cases, namely the case of fully open extensions as well as the case of $S_G=S_F$. In the latter case, as $F$ is a subnetwork of $G$,  the results on subnetworks in \cite{joshi-shiu-II} apply.

\begin{corollary}\label{cor:main}
Consider a reaction network $F$ with kinetics $K_F$ and species set $\mX$ of cardinality $n$, and the following two cases:
\begin{itemize}
\item[(1)] $G$ is the fully open extension of $F$.
\item[(2)] $G$ is a partially open extension of $F$ with $\mX^\iota$ and $\mX^o$ consisting only of non-conserved species of $F$, hence $S_F=S_G$.
\end{itemize}
For both cases, conclusions (i), (ii) and (iii) of Theorem~\ref{thm:main} hold. 
\end{corollary}
\begin{proof}
We verify that  assumption (a), (b) and (c) of Theorem~\ref{thm:main} hold in  the two cases. 

(1) (a) and (c) hold trivially as \edit{$\mX^\iota= \mX^o$}. We have \edit{$S_{F,\mX\setminus \mX^o}^\perp=\emptyset$} and  (b) holds.  

(2) \edit{Assumptions  (a) and (c) hold trivially.   (b) holds since  $ S_{F, \mX^o}^\perp=\emptyset$.}
\end{proof}

Before giving the proof of Theorem~\ref{thm:main}, we illustrate it with several examples. 
Examples ~\ref{ex:inflows_outflows_2} and \ref{ex:2site} illustrate further that the conditions of Theorem~\ref{thm:main} cannot easily be relaxed.

 \begin{example}\label{ex:inflows_outflows_1}
  Consider the reaction network in Example~\ref{ex:main}, equipped with mass-action kinetics
\[ F\colon \quad X_1+X_4 \ce{<=>} X_2  \ce{<=>} X_3+X_4,\]
and the following partially open extension of $F$: 
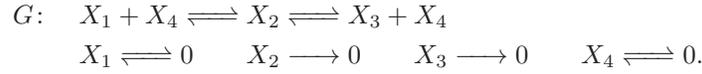
  \begin{align*}
   G\colon\quad & X_1 +X_4 \ce{<=>} X_2  \ce{<=>} X_3+X_4\\
                & X_1 \ce{<=>} 0 \qquad X_2 \ce{->} 0 \qquad X_3 \ce{->} 0 \qquad X_4 \ce{<=>} 0.
  \end{align*}
Here $\mX^o=\mX$ and $\mX^\iota=\{X_1,X_4\}$. 
Let us verify that $F,G$ satisfy the assumptions of Theorem~\ref{thm:main}.  
A matrix of conservation laws of $F$ is \edit{(reference to the species is added at the top)}:
\edit{\[ W_F=
\begin{blockarray}{cccc}
 X_1 & X_2 & X_3 & X_4 \\
\begin{block}{(cccc)}
         1 & 1 & 1 & 0 \\ 0 & 1 & 0 & 1 \\
 \end{block}
     \end{blockarray}.\] }%
     
          \vspace{-0.5cm}
  We write $W_F$ in the form \eqref{eq:block_diag_form_inflows_outflows} by choosing $W_1=W_F$ and $W_2$ empty, and check the conditions of Theorem~\ref{thm:main}. 
As $\mX^\iota\subseteq \mX^o=\mX$, assumptions  (a) and (b) hold.  Finally, the cone $\mathcal{C}^o(W_F)$ is generated
by the first and last columns of $W_F$. Hence, the assumptions of Theorem~\ref{thm:main} hold, and the conclusions \edit{(i)--(iii)} regarding lifting properties of the steady states of $F$ to $G$ apply. 

If $F$ models a reversible enzymatic reaction, then the extension $G$ includes degradation of every species, and assimilation of the substrate $X_1$ and the enzyme $X_4$  from the external environment.

Observe that as $\mathcal{C}^o(W_F)$ also is generated by the third and fourth columns of $W_F$, Theorem~\ref{thm:main} also applies to the partially open extension with inflow set \edit{$\mX^\iota=\{X_3,X_4\}$}. 
 \end{example}
 
 \begin{example}(Hybrid histidine kinase)\label{ex:inflows_outflows_2}
We consider  a well-studied simplified model of a hybrid histidine kinase HK with two phosphorylated sites and transference of  the phosphate group to a histidine phosphotransferase Hpt \cite{feliu:unlimited}. The reaction network $F$
is
  \begin{align*}
   F\colon\quad & {\rm HK}_{00} \ce{->} {\rm HK}_{p0} \ce{->}{\rm HK}_{0p} \ce{->} {\rm HK}_{pp}   \qquad  {\rm Hpt}_{p} \ce{->} {\rm Hpt}_{0}
\\ & 
 {\rm HK}_{pp} + {\rm Hpt}_{0}   \ce{->} {\rm HK}_{p0}+{\rm Hpt}_{p} \qquad {\rm HK}_{0p}+{\rm Hpt}_{0} \ce{->} {\rm HK}_{00}+{\rm Hpt}_{p}.              
  \end{align*}
Assuming mass-action kinetics, this network can have one or three positive non-degenerate steady states, depending on the choice of reaction rate constants and stoichiometric compatibility class \cite{feliu:unlimited}. If it has three positive non-degenerate steady states, then two of them are exponentially stable and the other exponentially unstable. If it  has only one, then it is exponentially stable \cite{torres:bistability}.

Here we explore different partially open extensions that satisfy the assumptions of Theorem~\ref{thm:main}. 
To this end, we consider the following matrix of conservation laws for $F$, \edit{where the species set is ordered as indicated by the column labels}:
\edit{\[W_F=\begin{blockarray}{cccccc}
{\rm HK}_{00} & {\rm HK}_{p0} & {\rm HK}_{0p} & {\rm HK}_{pp} & {\rm Hpt}_{0} & {\rm Hpt}_{p} \\
\begin{block}{(cccccc)}
         1 & 1 & 1 & 1 & 0 & 0 \\ 
         0 & 0 & 0 & 0 & 1 & 1\\
\end{block} 
     \end{blockarray}.\]}%
     
               \vspace{-0.5cm}
It follows that the sets \edit{$\mX^o_1 = \{ {\rm HK}_{00} , {\rm HK}_{p0} ,{\rm HK}_{0p} , {\rm HK}_{pp} \}$ and $\mX^o_2=\{{\rm Hpt}_{0} , {\rm Hpt}_{p}\}$} both satisfy assumption \edit{(b)} of Theorem~\ref{thm:main}. With $\mX^\iota_1=\{{\rm HK}_{00} \}$ or $\mX^\iota_2=\{{\rm Hpt}_{0} \}$,   assumptions \edit{(a)} and (c) of Theorem~\ref{thm:main} hold for both \edit{pairs of sets indicated by the subindex}. We obtain the following two partially open extensions
  \begin{align*}
   G_1\colon\quad &
 {\rm HK}_{00} \ce{->} {\rm HK}_{p0} \ce{->}{\rm HK}_{0p} \ce{->} {\rm HK}_{pp}   \qquad  {\rm Hpt}_{p} \ce{->} {\rm Hpt}_{0}
\\ & 
 {\rm HK}_{pp} + {\rm Hpt}_{0}   \ce{->} {\rm HK}_{p0}+{\rm Hpt}_{p} \qquad {\rm HK}_{0p}+{\rm Hpt}_{0} \ce{->} {\rm HK}_{00}+{\rm Hpt}_{p} \\
&   {\rm HK}_{00}   \ce{<=>} 0 \qquad  {\rm HK}_{p0} \ce{->} 0 \qquad {\rm HK}_{0p} \ce{->} 0  \qquad {\rm HK}_{pp}
 \ce{->} 0 
\\[6pt]
   G_2\colon\quad &  {\rm HK}_{00} \ce{->} {\rm HK}_{p0} \ce{->}{\rm HK}_{0p} \ce{->} {\rm HK}_{pp}   \qquad  {\rm Hpt}_{p} \ce{->} {\rm Hpt}_{0}
\\ & 
 {\rm HK}_{pp} + {\rm Hpt}_{0}   \ce{->} {\rm HK}_{p0}+{\rm Hpt}_{p} \qquad {\rm HK}_{0p}+{\rm Hpt}_{0} \ce{->} {\rm HK}_{00}+{\rm Hpt}_{p} \\
&   {\rm Hpt}_{0}   \ce{<=>} 0 \qquad  {\rm Hpt}_{p} \ce{->} 0.
  \end{align*}
We conclude that for appropriate choices of flow reaction rate constants,  these networks admit three positive steady states in some stoichiometric compatibility class, of which two are exponentially stable and the other is exponentially unstable.  
Network $G_1$ models the situation in which all phosphoforms of HK are degraded or exit the system, and the non-phosphorylated form is synthesized. Similarly, network $G_2$ models the situation in which both phosphoforms of Hpt are degraded, but only the non-phosphorylated form is synthesized or enters the system. 

Adding inflow reactions for the species in $\mX^o_1$ or $\mX_2^o$ does not alter the conclusion. Additionally, by joining  the inflow and outflow sets of both extensions, we obtain a new extension where Theorem~\ref{thm:main} also applies. 

We investigate what happens when flow reactions for \edit{Hpt$_0$}, Hpt$_p$ are added in ways that do not satisfy  assumptions \edit{(a)--(c)}. By analyzing the resulting systems in detail \edit{with mathematical software like Maple, or using for instance the CRNT Toolbox \cite{crnttoolbox}}, we easily see that if the outflow from Hpt$_{p}$ in $G_2$ is removed or if only outflow from \edit{Hpt$_0$} and input to Hpt$_{p}$  are considered, then the network has at most one positive steady state. If the outflow from \edit{Hpt$_0$} is removed or if only inflow and outflow reactions for Hpt$_{p}$ are considered, then the network has at most two positive steady states (\edit{as the steady state equations reduce to a degree two polynomial equation in the concentration of HK$_{00}$}). 
Hence, in all four cases, the conclusions of Theorem~\ref{thm:main} do not hold.
 \end{example}

\begin{example}(Double phosphorylation cycle) \label{ex:2site}
We consider a double phosphorylation cycle comprising a substrate $S$ with two ordered phosphorylation sites admitting  three phosphoforms $S_0,S_1,S_2$ with none, one, or two phosphate groups attached respectively. We assume phosphorylation and dephosphorylation are enzyme mediated and proceed in a sequential and distributive way. This gives rise to the following reaction network:
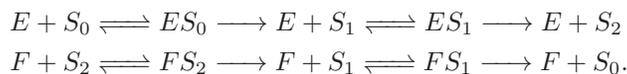
\begin{align*}
E+S_0 & \ce{<=>} ES_0 \ce{->}  E+ S_1  \ce{<=>} ES_1 \ce{->}  E+ S_2 \\
F+S_2 & \ce{<=>} FS_2 \ce{->}  F+ S_1  \ce{<=>} FS_1 \ce{->}  F+ S_0.
\end{align*}
Under mass-action kinetics, this network is known to admit up to three positive non-degenerate steady states \cite{Wang:2008dc}, as well as parameter choices for which there are two exponentially stable positive steady states and one exponentially unstable positive steady state \cite{rendall-2site}.

A matrix of conservation laws   is
\edit{\[W_F=\begin{blockarray}{ccccccccc}
E & F & S_0 & S_1 & S_2 & ES_0 & ES_1 & FS_2& FS_1 \\
\begin{block}{(ccccccccc)} 
         1 & 0 & 0 & 0 & 0  & 1 & 1 & 0 & 0 \\ 
         0 & 1 & 0 & 0 & 0 & 0 & 0  & 1 & 1 \\
         0 & 0 & 1 & 1 & 1 & 1 & 1 &1 & 1 \\
\end{block} 
     \end{blockarray}.\]}%
     
               \vspace{-0.5cm}
The assumptions of Theorem~\ref{thm:main} hold for $\mX^o=\mX$ and $\mX^\iota=\{ E, F, S_0\}$. Hence, in particular, bistability arises for this partially open extension. This implies that we need degradation of all species, but only production of  $E,F,S_0$.

In order to obtain even smaller sets of inflow and outflow reaction that preserve bistability, we can consider the following reduced network 
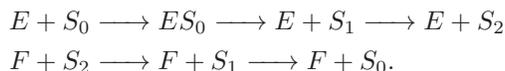
\begin{align*}
E+S_0 & \ce{->} ES_0 \ce{->}  E+ S_1   \ce{->}  E+ S_2 \\
F+S_2 &  \ce{->}  F+ S_1  \ce{->}  F+ S_0.
\end{align*}
This network admits three positive non-degenerate steady states, and whenever this is the case, two of them are exponentially stable and the other unstable \cite{torres:bistability}.  The original double phosphorylation network is obtained by the addition of the intermediates $ES_1,FS_2,FS_1$ and after making binding reactions reversible. These two modifications are known to preserve the number and stability of the steady states \cite{joshi-shiu-II,feliu:intermediates}.
This reduced network admits the following matrix of conservation laws:
\edit{\[\begin{blockarray}{cccccc}
E & F & S_0 & S_1 & S_2 & ES_0 \\
\begin{block}{(cccccc)} 
         1 & 0 & 0 & 0 & 0  & 1  \\ 
         0 & 1 & 0 & 0 & 0 & 0  \\
         0 & 0 & 1 & 1 & 1 & 1   \\
\end{block} 
     \end{blockarray}.\]}%
Hence the hypotheses of Theorem~\ref{thm:main} hold for $\mX^o=\{E,S_0, S_1,S_2 , ES_0\}$ and $\mX^\iota=\{ E, S_0\}$. In particular, bistability arises for this partially open extension. 
We proceed now to add  the intermediates $ES_1,FS_2,FS_1$ and  make binding reactions reversible to obtain the following network:
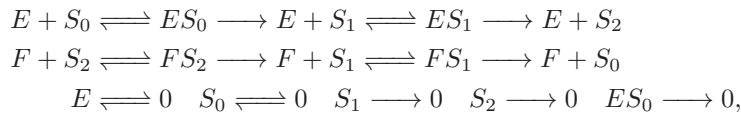
\begin{align*}
E+S_0 & \ce{<=>} ES_0 \ce{->}  E+ S_1  \ce{<=>} ES_1 \ce{->}  E+ S_2 \\
F+S_2 & \ce{<=>} FS_2 \ce{->}  F+ S_1  \ce{<=>} FS_1 \ce{->}  F+ S_0 \\
 E & \ce{<=>}  0 \quad  S_0  \ce{<=>}  0 \quad  S_1  \ce{->}  0  \quad  S_2  \ce{->}  0  \quad  ES_0  \ce{->}  0,
\end{align*}
which also admits three positive non-degenerate steady states, two of which are exponentially stable. 
By combining Theorem~\ref{thm:main} with previously known operations that preserve bistability,  we have obtained a smaller partially open extension of the double phosphorylation cycle that also admits bistability.

Actually, the outflow set can be made even smaller. 
The partially open extensions given by $\mX^o=\{ S_0, S_1,S_2\}$ and $\mX^\iota=\{S_0\}$ or  $\mX^o=\mX^\iota=\{E\}$ also admit three positive non-degenerate steady states,  \edit{as it can be derived by using for instance the CRNT Toolbox \cite{crnttoolbox} or  the method in \cite{FeliuPlos}. (See also \cite{joshi-shiu-III} for other methods.)}
The sets $\mX^o, \mX^{\iota}$ do not satisfy the assumptions of Theorem~\ref{thm:main}, as there are no conservation laws with support in $\mX^o$, and it is not evident how it could follow from Theorem~\ref{thm:main} after first reducing the network as we did above. 
\edit{This shows that the conditions of Theorem~\ref{thm:main} are only sufficient and not necessary for the lifting properties. It does not seem straightforward to derive more general conditions that would include this example. For instance, one might think that it is enough to consider as outflow species all the species that are not enzyme-substrate complexes in one or more conservation laws, and choose one species per conservation law to additionally be an inflow species. But} the partially open extension with respect to the sets $\mX^o=\mX^\iota=\{E,F\}$  does not admit three positive non-degenerate steady states, \edit{hence contradicting this hypothesis}. This shows the subtleties in obtaining general results with respect to how to lift properties of steady states of a network to partially open extensions. 
\end{example}

\section{Proof of Theorem~\ref{thm:main}}\label{sec:proof}

This section is devoted to \edit{the proof of}  Theorem~\ref{thm:main}. 
\edit{Let $n_1$ be the cardinality of $\mX^\iota\cup \mX^o$ and $n_{1,o}$ the cardinality of $\mX^o$. We order 
the set of species such that the species in $\mX^o$ are first, then we have the species in $\mX^\iota$ not in $\mX^o$, and finally the species in $\mX \setminus (\mX^\iota \cup \mX^o)$.} 
Consider a matrix $W_F$ of conservation laws of $F$ of the following form \eqref{eq:block_diag_form_inflows_outflows}:
\[ W_F=\begin{blockarray}{ccc}
& n_1 & n_2  \\ 
\begin{block}{c(cc)}
d_1 &   W_1 & 0 \\
d_2 &   0     & W_2 \\
\end{block}  \end{blockarray} \in \R^{d\times n}, \] 
\edit{where throughout block matrix labels indicate the number of rows and columns of the corresponding blocks. }%
Such a matrix exists by assumption \edit{(b)} (c.f. Remark~\ref{rk:matrix}). 
\edit{By assumption (a), the last $n_1-n_{1,o}$ columns of $W_1$ are zero. Note that as opposed to Remark~\ref{rk:matrix}, the block matrix $W_1$ includes the zero columns corresponding to inflow species in $\mX^\iota\setminus \mX^o$. This makes the notation in this proof lighter.}
Consider the stoichiometric matrix $N_F\in \R^{n\times m}$ of $F$ of rank $s$, and  write it in block form  as
 \[ N_F=\begin{pmatrix} 
      N_1 \\
      N_2
     \end{pmatrix},\qquad N_1\in\R^{n_1\times m}, \ N_2\in\R^{n_2\times m}.\]
 Let $s_1$ and $s_2$ be the rank of $N_1$ and $N_2$, respectively.
Fix two matrices $N_1^\prime\in \R^{s_1\times m}$ and $N_2^\prime\in \R^{s_2\times m}$ of rank $s_1,s_2$ and such that $\ker(N_1')=\ker(N_1)$,  $\ker(N_2')=\ker(N_2)$. Then  the matrix
\[ N'_F=\begin{pmatrix} 
      N_1' \\
      N_2'
     \end{pmatrix}\] 
satisfies $\ker(N'_F)=\ker(N_F)$. 
Since the rows of $W_1$ and $W_2$ generate the left kernel of  $N_1$ and $N_2$ respectively, we have
\[ s_1+s_2=n_1-d_1+n_2-d_2=n-d=s. \] 
It follows that $N'_F$ has rank $s$. \edit{Let  $A\in \R^{s_1\times n_1}$ of full rank $s_1$ such that $N_1'=A N_1$.}

\medskip
Assume $F$ has $\ell$ non-degenerate positive steady states $c_1,\dots,c_\ell$ in the stoichiometric compatibility class with equations $W_F x = T$ with $T\in \R^d$, and let $C\subseteq \R^n_{>0}$ be an open subset  containing the steady states.  The steady state equations for $F$ in the given stoichiometric compatibility class are
\[ N'_F K_F(x)= 0,\qquad W_F x = T,\]
which are equivalent to
\begin{equation}\label{eq:eqF}
N_1' K_F(x)=0,\quad N_2' K_F(x)=0,\quad W_F x = T.
\end{equation}

Let \edit{$E_O= \begin{pmatrix}
{\rm id}_{n_{1,o}\times n_{1,o}} & 0 \\ 0 & 0
\end{pmatrix}\in\R^{n_{1}\times n_1}$} be the matrix with zero entries except for the diagonal entries $(i,i)$ for $X_i\in \mX^o$, which are equal to one. Hence $W_1 E_O=W_1$. 
Similarly, let $E_I\in\R^{n_1\times n_1}$ be the  matrix with zero entries except a one in entry $(i,i)$ if  $X_i\in \mX^\iota$. 

Let $T_1,T_2$ be the \edit{conserved quantities} corresponding to the conservation laws given by $W_1,W_2$, respectively, that is, $T=(T_1,T_2)$. 
 Let $\pi_1\colon \R^n\rightarrow\R^{n_1}$ and $\pi_2\colon \R^n\rightarrow\R^{n_2}$ be the projections on the first $n_1$ components and on the last $n_2$ components, respectively. 
Choose \edit{$\widehat{x}\in \R^{n}_{>0}$} with support in $\mX^\iota$ and such that \edit{$(W_1 \  0) \widehat{x}=T_1$, or, equivalently $W_1 \pi_1(\widehat{x})=T_1$}. Such an $\widehat{x}$ exists by  assumption (c) of Theorem~\ref{thm:main}. 
Note that since the support of $\widehat{x}$ is in $\mX^\iota$, we have \edit{ $W_1 \pi_1(\widehat{x})=W_1 E_I \pi_1(\widehat{x})$. }
With this choice, consider now the map
\[ H\colon \R\times C \rightarrow \R^{s_1}\times \R^{s_2}\times \R^{d_1}\times \R^{d_2}\equiv \R^{n} \]
defined by
\begin{align*}
H(\theta,x) &  = \big(\ N_1'K_F(x) - \theta A E_O \pi_1(x) + \theta A E_{I}\edit{\pi_1(\widehat{x})},\ N_2^\prime K_F(x), \\ &  \hspace{3cm}\  W_1 \pi_1(x) -   T_1, \  W_2 \pi_2(x) -   T_2\ \big).
\end{align*}
The function $H$ is $\C^1$ and we will prove statement (i) using the Implicit Function Theorem on $H$.
For that, first note that when $\theta=0$, the equation $H(0,x)=0$ amounts to \eqref{eq:eqF}. Hence, $c_1,\dots,c_\ell\in C$ satisfy the equation $H(0,x)=0$. 
Moreover,
\edit{  \begin{equation}\label{eq:derivative_in_proof_inflows_outflows}
   \partial_x H(\theta, x)=\begin{blockarray}{ccc}
n_1 & n_2 &  \\ 
\begin{block}{(cc)c}
                            N_1'\partial_{\pi_1(x)}K_F(x) - \theta A E_O  & N_1'\partial_{\pi_2(x)}K_F(x) & s_1 \\
                            N_2'\partial_{\pi_1(x)}K_F(x)                     & N_2'\partial_{\pi_2(x)}K_F(x) & s_2 \\
                            W_1                                          & 0 & d_1 \\
                            0                                            & W_2 & d_2 \\
\end{block}                           \end{blockarray}.
  \end{equation}}
  In particular, for $\theta=0$ we have
  $$\partial_x H(0, x)=\begin{pmatrix}
                          N_1'\partial_x K_F(x) \\
                          N_2'\partial_x K_F(x) \\
                          W_F
                         \end{pmatrix} =\begin{pmatrix}
                          N'_F \partial_x K_F(x) \\
                          W_F
                         \end{pmatrix}.$$
This matrix is non-singular when evaluated at $c_1,\dots,c_\ell$ as the steady states are non-degenerate by assumption, see Lemma~\ref{lemma:nondeg}.
For each $i=1,\dots,\ell$, we apply the Implicit Function Theorem to the point $(0,c_i)\in (-\varepsilon,\varepsilon)\times C$ and the function $H$, to conclude that there exists an open interval $I_i\subseteq (-\varepsilon,\varepsilon)$ containing $0$, an open set $U_i\subseteq C$ containing $c_i$, and a differentiable function
$$ h_i\colon I_i \rightarrow U_i,$$
such that for all $\theta\in I_i$, $H(\theta,h_i(\theta))=0$ and $h_i(0)=c_i$. Further, $I_i$ can be chosen small enough such that  the map $\partial_{x} H(\theta,h_i(\theta))$ is non-singular for every $\theta\in I_i$, since it is non-singular at $\theta=0$.
Since all points $c_i$ are distinct, there exist pairwise disjoint open sets $V_i\subseteq C $ containing $c_i$. 
We redefine $U_i$ to be $U_i\cap V_i$, which contains $c_i$, and $I_i$ to be the connected component of the anti-image of $U_i$ by $h_i$ that contains $0$. With these definitions, the images of the maps $h_i\colon I_i\rightarrow U_i$ are pairwise disjoint and  the components of $h_i(\theta)$ are positive.

Consider the open interval  $I=\bigcap_{i=1}^\ell I_i$, which contains $0$. All maps $h_i$ are defined on $I$. If $\theta\in I$, then by construction
\[H(\theta,h_i(\theta))=0 \]
for all $i=1,\dots,\ell$ and all $h_i(\theta)$ are  distinct. 

Part (i) of the theorem will follow if we show that for $\theta$ small enough and positive,  $h_1(\theta),\dots,h_\ell(\theta)$ are positive non-degenerate steady states of $G$ for  a choice of reaction rate constants of the flow reactions, and that they belong to the  stoichiometric compatibility class $x_0+S_G$.

So fix $\theta>0$ with $\theta\in I$. 
\edit{Consider the matrix } 
\[  \widehat{N}_G=\begin{pmatrix}
          N_1  & -E_O &  E_I \\
          N_2 & 0   &  0
         \end{pmatrix} \in \R^{(n_1+n_2)\times \edit{(m+2n_1)}}, \]  
and the matrix 
  \[ \widehat{N}'_G=\begin{pmatrix}
          N_1  &  -E_O &  E_I \\
          N'_2 & 0   & 0
         \end{pmatrix} \in \R^{(n_1+s_2)\times \edit{(m+2n_1)}}\]
\edit{By construction, the stoichiometric matrix $N_G$ of $G$ is
obtained by removing the zero columns  in the second (resp. third) column blocks of $\widehat{N}_G$ corresponding to the species not in $\mX^o$ (resp. $\mX^\iota$). We define $N'_G$ from $\widehat{N}'_G$ analogously. 
Note $\ker(\widehat{N}_G)=\ker(\widehat{N}'_G)$ and also $\ker(N_G)=\ker(N'_G)$ . } By Lemma~\ref{lem:directsum} \edit{and assumption (a)}, the matrix $W_G=(0 \ \ W_2)$ is a matrix of conservation laws for $G$. Hence the rank of $N_G$ is $n-d_2=n_1+s_2$. This implies that $N_G'$ has maximal rank. 
Further, the stoichiometric compatibility class of $G$ containing $c_1,\dots,c_\ell$ has equations 
\begin{equation}\label{eq:WG}
W_G x = T_2 \edit{\qquad\textrm{or equivalently}\qquad W_2 \pi_2(x)=T_2.}
\end{equation}

 Let $K_G^\theta$ be the kinetics of $G$ agreeing with $K_F$ for the common reactions, such that the reaction rate constant of $X_j\to 0$ is $\theta$ if $X_j\in \mX^o$, and the reaction rate constant of $0\to X_j$ is $\theta\widehat{x}_j$ for $X_j\in \mX^\iota$ (where $\widehat x$ is as defined above).
Then 
\[ N_G' K^\theta_G(x) = \begin{pmatrix}  N_1 K_F(x) - \theta E_O \pi_1(x) + \theta E_I  \edit{\pi_1(\widehat{x})} \\ N_2^\prime K_F(x)\end{pmatrix} \in \R^{n_1+s_2}.\]
Hence, the steady states of $G$ in the class defined by \eqref{eq:WG} are the solutions to the equations
  \begin{gather}
   \label{eq:steady_states_opened_network_first_block}
   N_1 K_F(x) - \theta E_O \pi_1(x) + \theta E_I \edit{\pi_1(\widehat{x})} =0,\\ 
   \label{eq:steady_states_opened_network_second_block}
   \edit{W_2 \pi_2(x)=T_2},\qquad N_2^\prime K_F(x)=0.
  \end{gather}

Consider the matrix
\[ P= \begin{pmatrix}
W_1 \\ A    \end{pmatrix} \in \R^{n_1\times n_1},\]
where   $A\in \R^{s_1\times n_1}$ \edit{was fixed above to be of full rank $s_1$ and} such that $N_1'=A N_1$. 
 The matrix $P$ is invertible since $\ker(P)=0$. To see this, note that $W_1z=0$ implies $z$ belongs to the column span of $N_1$, that is $z=N_1 y$ for $y\in \R^{m}$. Then, $0=A\, z =AN_1 y = N_1' y$. As $\ker (N_1')=\ker (N_1)$, this implies $z=0$. 

 Hence, equation~\eqref{eq:steady_states_opened_network_first_block} is equivalent to
\[  PN_1 K_F(x) - \theta P(E_O \pi_1(x)) + \theta P E_I \edit{\pi_1(\widehat{x})} =0. \] 
By construction, as $W_1 E_O=W_1$ \edit{and $W_1N_1=0$}, we have
\begin{align*}
PN_1 K_F(x) & = \begin{pmatrix}
                   0 \\ N_1'K_F(x)
                  \end{pmatrix},\\
P(E_O \pi_1(x)) & = \begin{pmatrix}  W_1 \pi_1(x)  \\  A E_O \pi_1(x) \end{pmatrix}, \\
 P E_I \edit{\pi_1( \widehat{x})}  & =   \begin{pmatrix}    W_1 E_I \edit{\pi_1( \widehat{x})}     \\  A E_I \edit{\pi_1( \widehat{x})}       \end{pmatrix} = \begin{pmatrix}    T_1  \\  A E_I \edit{\pi_1( \widehat{x})}      \end{pmatrix}. 
 \end{align*}
  Then, equation~\eqref{eq:steady_states_opened_network_first_block} holds if and only if
  \begin{equation}\label{eq:out}
   N_1'K_F(x) - \theta A E_O \pi_1(x) + \theta AE_I\edit{\pi_1( \widehat{x})}   =0,\qquad  
 W_1 \pi_1(x) -   T_1 =0.
  \end{equation}
 Now, let $c^*\in C$ be such that $H(\theta, c^*)=0$. Then by definition  \edit{of} $H$, both the equations in \eqref{eq:out}, which are equivalent to \eqref{eq:steady_states_opened_network_first_block},  and the equations \eqref{eq:steady_states_opened_network_second_block} hold. 
It follows that $c^*$ is  a positive steady state of $G$ for the kinetics $K^\theta_G$ in the class defined by \eqref{eq:WG}. 

This gives that $h_1(\theta),\dots,h_\ell(\theta)$ define $\ell$ positive steady states of $G$ for $\theta\in I$ and positive. Using that $\partial_{x} H(\theta,h_i(\theta))$ is non-singular, we prove  that these steady state also are non-degenerate provided $\theta\in I$ and positive. 
For this, fix $c^*\in C$ such that $H(\theta,c^*)=0$ and satisfying that $\partial_xH(\theta,c^*)$ is non-singular. By Lemma~\ref{lemma:nondeg}, $c^*$ is non-degenerate if and only if
  $$J =\begin{pmatrix}
                            N_1\partial_{\pi_1(x)}K_F(c^*) - \theta E_O & N_1\partial_{\pi_2(x)}K_F(c^*)\\
                            N'_2\partial_{\pi_1(x)}K_F(c^*)             & N'_2\partial_{\pi_2(x)}K_F(c^*)\\
                            0                                      & W_2
                           \end{pmatrix}$$
  is non-singular. 
Note that
\[ \begin{pmatrix}
     P & 0\\
     0 & \text{id}_{n_2}
    \end{pmatrix}J =
    \begin{pmatrix}
   -   \theta W_1                               & 0\\
     N'_1\partial_{\pi_1(x)}K_F(c^*) - \theta AE_O & N'_1\partial_{\pi_2(x)}K_F(c^*)\\
     N'_2\partial_{\pi_1(x)}K_F(c^*)               & N'_2\partial_{\pi_2(x)}K_F(c^*)\\
     0                                        & W_2
    \end{pmatrix}.\] 
As $\begin{pmatrix}
     P & 0\\
     0 & \text{id}_{n_2}
    \end{pmatrix}$ is invertible, and the matrix on the right-hand side is non-singular by hypothesis and  \eqref{eq:derivative_in_proof_inflows_outflows}, we conclude that $J$ is non-singular. This shows that $c^*$ is non-degenerate. 

\medskip
This concludes the proof of statement (i). 
Statement (iii) follows from statement (ii).
So all we need is to show statement (ii). 
Consider a steady state $c_i$ of $F$ and the corresponding steady state $h_i(\theta)$ of $G$ as above. Let $J_F$ be the Jacobian of the species formation rate function of $F$ evaluated at $c_i$ and $J_G$  the Jacobian of the species formation rate function of $G$ evaluated at $h_i(\theta)$. 
Then 
\[ J_F =\begin{pmatrix}
                          N_1\partial_x K_F(c_i) \\
                          N_2\partial_x K_F(c_i) 
                                                   \end{pmatrix}, \qquad   J_G =\begin{pmatrix}
                          N_1\partial_x K_F(h_i(\theta)) - \theta \Big(E_O \quad  0_{n_1\times n_2} \Big) \\
                          N_2\partial_x K_F(h_i(\theta)) 
                                                 \end{pmatrix}, \]
where $0_{n_1\times n_2}$ is the zero matrix of size $n_1\times n_2$.  
Then $J_F$ has the eigenvalue $0$ with multiplicity at least $d=d_1+d_2$, and assume further that it has $r_1$ eigenvalues with positive real part, and $r_2$ eigenvalues with negative real part. As $h_i(0)=c_i$,  $J_G$  also has $r_1$ eigenvalues with positive real part, and $r_2$ eigenvalues with negative real part for $\theta>0$ small enough.
All that remains is to show that $J_G$ has $d_1$ additional eigenvalues with negative real part. 

We consider left  eigenvectors of $J_F, J_G$ for convenience. 
Then, as $N_1$ has rank $s_1$, $d$ distinct linearly independent left eigenvectors $u_1,\dots,u_d$ with eigenvalue $0$ of $J_F$ can be chosen such that for $j=1,\dots,d_1$ we have $\pi_1(u_j)^t N_1=0$, $\pi_2(u_j)=0$ (where superscript $t$ denotes the transpose vector). That is, $\pi_1(u_j)$, for $j=1,\dots,d_1$, form a basis of the left kernel of $N_1$.  In particular, $\pi_1(u_j)$  belongs to the row span of $W_1$: $\pi_1(u_j)= v_j^t W_1$ for $v_j\in \R^{d_1}$. Then the equality $v_j^t  W_1 E_O= v_j^t  W_1$ gives $\pi_1(u_j)^t E_O =  \pi_1(u_j)^t$. 
Now, for $j=1,\dots,d_1$, we have
\begin{align*}
u_j^t J_G   & = u_j^t  \begin{pmatrix}   N_1\partial_x K_F(h_i(\theta)) - \theta \Big( E_O \quad  0_{n_1\times n_2}\Big) \\
                          N_2\partial_x K_F(h_i(\theta)) 
                                                   \end{pmatrix}  \\
&  =  \pi_1(u_j)^t N_1\partial_x K_F(h_i(\theta)) +  \pi_2 (u_j)^t N_2\partial_x K_F(h_i(\theta))   - \theta\, \pi_1(u_j)^t  \Big( E_O \quad 0_{n_1\times n_2}\Big)  \\ & = 
                           - \theta\,  \Big( \pi_1(u_j)^t  E_O \quad 0_{1\times n_2} \Big)    = -\theta   u_j^t,
\end{align*}
 where in the last step we use that $\pi_2(u_j)=0$.
This shows that  $-\theta$ is an eigenvalue of $J_G$ with multiplicity at least $d_1$, completing the proof of statement (iii) and thereby the proof of Theorem~\ref{thm:main}. 
 \begin{flushright}
\qed
\end{flushright}

\section*{Acknowledgments}
 EF and CW acknowledge funding from the Independent Research Fund of Denmark.  Part of this work was done while the authors visited the Isaac Newton Institute, Cambridge, UK and we are grateful for the support offered from the institute.  The project was initiated while DC was at the University of Copenhagen.

%

\end{document}